\documentclass[11pt,a4paper]{article}

\usepackage{amsmath,amssymb,amsthm}
\usepackage{url}
\usepackage[margin=2cm]{geometry}

\newtheorem{proposition}{Proposition}

\newtheorem{lemma}{Lemma}
\newtheorem{corollary}{Corollary}

\newcommand{\gf}{\mathrm{GF}}

\renewcommand{\dh}{d_{\mathrm{H}}}

\newcommand{\As}{A_{\mathrm S}}

\begin{document}
\title{Constant-Weight Array Codes}
\author{Maximilien Gadouleau\footnote{Department of Computer Science, Durham University, Durham DH1 3LE, UK. Email: \url{m.r.gadouleau@durham.ac.uk}}}
\date{\today}
\maketitle

\begin{abstract}
Binary constant-weight codes have been extensively studied, due to both their numerous applications and to their theoretical significance. In particular, constant-weight codes have been proposed for error correction in store and forward. In this paper, we introduce constant-weight array codes (CWACs), which offer a tradeoff between the rate gain of general constant-weight codes and the low decoding complexity of liftings. CWACs can either be used in the on-shot setting introduced earlier or in a multi-shot approach, where one codeword consists of several messages. The multi-shot approach generalizes the one-shot approach and hence allows for higher rate gains. We first give a construction of CWACs based on concatenation, which generalizes the traditional erasure codes, and also provide a decoding algorithm for these codes. Since CWACs can be viewed as a generalization of both binary constant-weight codes and nonrestricted Hamming metric codes, CWACs thus provide an additional degree of freedom to both problems of determining the maximum cardinality of constant-weight codes and nonrestricted Hamming metric codes. We then investigate their theoretical significance. We first generalize many classical bounds derived for Hamming metric codes or constant-weight codes in the CWAC framework. We finally relate the maximum cardinality of a CWAC to that of a constant-weight code, of a nonrestricted Hamming metric code, and of a spherical code. 
\end{abstract}

\section{Introduction} \label{sec:introduction}
Binary constant-weight codes have been widely studied \cite{BSSS90} due to their applications to codes for asymmetric channels \cite{CR79}, DC-free codes \cite{vTB89}, and spherical codes \cite{AVZ00}. Recently, their connections to constant-dimension codes \cite{XF09, ES09}, which were proposed for error and erasure correction in noncoherent network coding \cite{KK08}, renewed their interest. Constant-weight codes have also been proposed for error correction in store and forward \cite{GG10}. Although nonbinary constant-weight codes have also received some attention \cite{OS02, CL07}, we shall only consider binary constant-weight codes. These codes also have a strong theoretical significance, as bounds on the maximum cardinality of a constant-weight code also yield bounds on the cardinality of binary Hamming metric codes \cite{MS77, MR+77}. Also, the Johnsonian is a very structured mathematical object, known as an association scheme.

In \cite{GG10}, we showed how binary codes in general and constant-weight codes in particular can be used for error correction in store and forward. This approach generalizes the traditional approach of lifting a Hamming metric code (typically a Reed-Solomon code), and higher data rates can be achieved by using such constant-weight codes. However, the encoding and complexities when using general constant-weight codes are usually much higher than their counterparts for liftings of Reed-Solomon codes. In this paper, we introduce constant-weight array codes, which can be viewed as a tradeoff between data rate and complexity because the two extreme cases of CWACs are exactly liftings of Hamming metric codes and general constant-weight codes. Due to their particular structure, CWACs can also be viewed as generalized constant-weight codes. Therefore, we believe CWACs also have an important theoretical significance, which can already be seen through the bounds we derive in this paper.

The rest of the paper is organized as follows. Section \ref{sec:preliminaries} reviews some necessary background. Section \ref{sec:concatenation} provides a construction of CWACs based on concatenation, and designs a decoding algorithm for these codes. The relations between CWACs and other classes of codes are investigated in Section \ref{sec:relations}, while the maximum cardinality of CWACs is studied in Section \ref{sec:CWACs}.

\section{Preliminaries} \label{sec:preliminaries}
In this section, we review the following three classes of codes: constant-weight codes, Hamming metric codes, and spherical codes. We recall that the minimum distance of a code is the minimum distance over all pairs of distinct codewords, while the diameter of an anticode is the maximum distance between two codewords.

We denote the set of all vectors in $\gf(2)^n$ and weight $w$ as $J(n,w)$. A {\em constant-weight code} is defined as any subset of $J(n,w)$ with a given minimum Hamming distance, while a constant-weight anticode is also a subset of $J(n,w)$ with a prescribed diameter. Let $B(n,w,d)$ be the maximum cardinality of a binary constant-weight code of length $n$, weight $w$, and minimum Hamming distance $2d$. Similarly, let $\beta(n,w,\delta)$ be the maximum cardinality of a constant-weight anticode in $J(n,w)$ with diameter $2 \delta$. A lower bound on $\beta(n,w,\delta)$ is given by the number of subsets of cardinality $w$ which intersect a given subset of cardinality $w-\delta+2i$ in at least $w-\delta$ positions for any $0 \leq i \leq \delta$:
\begin{equation} \label{eq:beta}
	\beta(n,w,\delta) \geq \max_{0 \leq i \leq \delta} \sum_{d=0}^{\delta} \binom{w-\delta + 2i}{i- \frac{d-\delta}{2}} \binom{n-w+\delta-2i}{\frac{\delta+d}{2} - i}.
\end{equation}
The lower bound in (\ref{eq:beta}) is known to be tight provided $n$ is large enough \cite{EKR61} and is conjectured in \cite{Fra78} to be tight for all values.

A {\em nonrestricted Hamming metric code} is a nonempty subset of $[q]^n$, where $q \in \mathbb{N}$ and $[q] = \{0,1,\ldots,q-1\}$. Let $C(q,n,d)$ be the maximum cardinality of a code in $[q]^n$ with minimum Hamming distance $d$. Finally, let $\gamma(q,n,\delta)$ be the maximum cardinality of an anticode in $[q]^n$ with diameter $\delta$. Remarkably, the value of $\gamma(q,n,\delta)$ is determined in \cite{AK98} and given by
\begin{equation} \label{eq:gamma1}
	\gamma(q,n,\delta) = q^{\delta-2r}\sum_{i=0}^r \binom{n}{i} (q-1)^i,
\end{equation}
where $r$ is the largest integer such that $2r < \min \left\{\delta+1, 2\frac{n-\delta-1}{q-2} \right\}$.

A {\em spherical code} of dimension $n$ is a set of points on the unit sphere in $\mathbb{R}^n$ with the Euclidean norm. The maximum cosine $s$ of a spherical code is related to its minimum Euclidean distance $d_E$ by $s = 1 - \frac{d_E^2}{2}$. The maximum cardinality of a spherical code of dimension $n$ and maximum cosine $s$ is denoted as $\As(n,s)$. The exact value of $\As(n,s)$ is known for $s \leq 0$ and given by (see \cite{AVZ00})
\begin{eqnarray}
	\nonumber
	\As(n,s) &=& \left\lfloor 1 - \frac{1}{s} \right\rfloor, \quad \text{if } s \leq -\frac{1}{n}\\
	\nonumber
	\As(n,s) &=& n+1, \quad \text{if } -\frac{1}{n} \leq s < 0\\
	\nonumber
	\As(n,0) &=& 2n,
\end{eqnarray}
while bounds on $\As(n,s)$ for $s>0$ are given in \cite{Lev79, BDB96}.

\section{Constant-weight array codes} \label{sec:CWACs}

\subsection{Construction of CWACs by concatenation} \label{sec:concatenation}

The model reviewed in Section \ref{sec:introduction} identifies a message with a binary vector whose weight is equal to the number of packets in the message. Accordingly, we investigated using constant-weight codes for error control in store and forward in \cite{GG10}, where each codeword corresponds to a message. This approach falls short when the number of errors that typically occur on the network is very low. For instance, in order to correct one packet loss and one packet injection every ten transmissions, one need to use a constant-weight code with minimum distance at least $5$ for each transmission. Therefore, in this setting, we could correct one packet loss and one packet injection at every single transmission. In this section, we generalize the constant-weight code approach by considering codewords associated to sequences of messages. In the example considered above, we could then use codewords made of ten messages, and use a code with error correction capability of $2$. In general, using codes on several messages allows to adjust the error correction of the code more accurately to the number of errors produced by the network, hence increasing the data rate.

A {\em constant-weight array code} (CWAC) is defined as a nonempty subset of words in $J(m,w)^n$. Hence a constant-weight code in $J(m,w)$ can be viewed as a CWAC of length $1$. In this section, we design a class of good CWACs with low-complexity decoders. In Proposition \ref{prop:concatenation} below we construct a CWAC by concatenation, where the outer code is a nonrestricted Hamming metric code and the inner code is a constant-weight code. This can be viewed as a generalization of the lifting operation defined in \cite{GG10}.

\begin{proposition}[Construction of CWACs by concatenation]\label{prop:concatenation}
Let $\mathcal{I} \subseteq J(m,w)$ be a constant-weight code with minimum distance $2f$ and $\mathcal{O} \subseteq \mathcal{I}^n$ be a nonrestricted code with minimum distance $e$. Then the concatenated code $\mathcal{O} \circ \mathcal{I}$ is a CWAC in $J(m,w)^n$ with minimum distance at least $2ef$.
\end{proposition}

The proof of Proposition \ref{prop:concatenation} is straightforward and hence omitted.



The concatenation construction is well fitted for our purpose, as it ensures that each message is a codeword in a constant-weight code with minimum distance $2f$.

We first investigate when a modified Reed-Solomon code can be used as outer code. We remark that for $w=1$, the construction in Proposition \ref{prop:concatenation} reduces to the lifting operation. The lifting of a Reed-Solomon code requires that $n \leq m$, that is, the total weight $W = n$ and the total length $M = nm$ are related by $W \leq \sqrt{M}$. Therefore, this is a strong limitation on the number of packets in a message. On the other hand, we remark that when $w > 1$, the constraint on the parameter, given by $n \leq B(m,w,f)$ is loose when $f=1$, that is $A(m,n,w,d) \geq {m \choose w}^{n-d+1}$ for ${m \choose w} \geq n$.

We now study the decoding of CWACs obtained by concatenation. By \cite[Proposition A.1]{Gur02}, if there exists an algorithm running in time $t_i$ to uniquely decode $\mathcal{I}$ up to distance $f-1$ and if there exists an algorithm running in time $t_o$ to uniquely decode $\mathcal{O}$ from $t$ errors and $u$ erasures where $2t+u < e$, then there exists an algorithm running in $O(nt_i + ft_o)$ that uniquely decodes $\mathcal{O} \circ \mathcal{I}$ up to $ef-1$. First, let us assume that $\mathcal{I}$ has cardinality $2^{\lfloor \log_2 B(m,w,f) \rfloor}$. The algorithm to decode $\mathcal{I}$ is straightforward: enumerate all the vectors in $J(m,w)$ at distance at most $f-1$ from the received column and check if it is a codeword. The maximum complexity is on the order of $O[m \beta(n,w,f-1)]$. If $n \leq |\mathcal{I}|$, then the outer code can be a Reed-Solomon code (or a modified Reed-Solomon code, either shortened or extended). The receiver can hence use the Berlekamp-Massey algorithm, which has complexity on the order of $O[e^2 \log_2 B(m,w,f)]$ binary operations to correct the errors and erasures. The total complexity is hence on the order of $O[m \beta(n,w,f-1) + fe^2 \log_2 B(m,w,f)]$.

\subsection{Bounds on CWACs} \label{sec:bounds_CWACs}

Although constant-weight array codes were defined as sets of words in $J(m,w)^n$, by expanding each element of $J(m,w)$ as a column vector they can be equivalently defined as sets of matrices in $\gf(2)^{mn}$ with weight $w$ on any column. Constant-weight array codes can finally be defined as a subclass of constant-weight codes of length $M = mn$, constant-weight $W = nw$, and where all columns (the vectors in $J(m,w)$) have weight $w$. Therefore, CWACs are both a special class and a generalization of constant-weight codes. Since their minimum Hamming distance is equal to twice their minimum modified Hamming distance, the main scope of this section is the study of the maximum cardinality $A(m,n,w,d)$ of a CWAC in $J(m,w)^n$ with minimum Hamming distance $2d$.

First of all, let ${\bf J}_{m \times n}$ be the all-ones matrix in $\gf(2)^{m \times n}$. Then for any matrices ${\bf M}, {\bf N} \in J(m,w)^n$, we have ${\bf M} + {\bf J}, {\bf N} + {\bf J} \in J(m,m-w)^n$ and $\dh({\bf M} + {\bf J}, {\bf N} + {\bf J}) = \dh({\bf M}, {\bf N})$. Therefore, $A(m,n,w,d) = A(m,n,m-w,d)$ for all $n$, $m$, and $w$, so we assume $2w \leq m$ without loss of generality henceforth. 

The term $A(m,n,w,d)$ is a non-decreasing function of $m$, $n$, and a non-increasing function of $d$. The values of $A(m,n,w,d)$ can be readily determined for extreme values of $d$: $A(m,n,w,1) = {m \choose w}^n$, $A(m,n,w,d) = 1$ if $d > nw$, and $A(m,n,w,nw) = \left\lfloor \frac{m}{w} \right\rfloor$. 

The Singleton bound for Hamming metric codes can easily be extended to CWACs.
\begin{proposition}[Singleton bound for CWACs]
For all $d \geq w+1$, $$A(m,n,w,d) \leq A(m,n-1,w,d-w).$$
\end{proposition}

\begin{proof}
For any two ${\bf c}, {\bf d} \in J(m,w)$ we have $\dh({\bf c}, {\bf d}) \leq 2w$. For any two ${\bf M}, {\bf N} \in J(m,w)^n$ we denote the first $n-1$ columns as ${\bf M}'$ and ${\bf N}'$, respectively and their last coordinate as ${\bf m}_n$ and ${\bf n}_n$, respectively. We have $\dh({\bf M}, {\bf N}) = \dh({\bf M}', {\bf N}') + \dh({\bf m}_n, {\bf n}_n) \leq \dh({\bf M}', {\bf N}') + 2w$. Let $\mathcal{C}$ be a code in $J(m,w)^n$ with minimum distance $d > w$ and cardinality $A(m,n,w,d)$, then the code $\mathcal{C}' = \{{\bf M}' : {\bf M} \in \mathcal{C}\}$ has minimum distance $d - w$ and cardinality $|\mathcal{C}|$.
\end{proof}

%

The Johnson bound of type I \cite{Joh62} can be easily generalized to the case of CWACs. However, such a generalization is straightforward and the resulting bound is not usually tight. Therefore, we focus on the Johnson bounds of type II. Before determining the generalization of the Johnson bounds of type II \cite{Joh62} to CWACs, we study the maximum cardinality $\alpha(m,n,w,\delta)$ of an antiCWAC in $J(m,w)^n$ with diameter $\delta$. First, by symmetry we have $\alpha(m,n,w,\delta) = \alpha(m,n,m-w,\delta)$ hence we consider $w \leq m$. Some special values include $\alpha(m,n,w,0) = 1$ and $\alpha(m,n,w,nw) = {m \choose w}^n$. We give a construction of antiCWACs in Lemma \ref{lemma:concatenation_anti} by concatenating a constant-weight anticode and a Hamming metric anticode. This, in turn, yields a lower bound on the maximum cardinality of antiCWACs.

\begin{lemma}[Construction of a CWAC anticode by concatenation] \label{lemma:concatenation_anti}
Let $\mathcal{I} \subseteq J(m,w)$ be a constant-weight anticode with diameter $\phi$ and $\mathcal{O} \subseteq \mathcal{I}^n$ be an anticode with diameter $\phi$, then the concatenated code $\mathcal{O} \circ \mathcal{I}$ is an $(m,n,w)$ antiCWAC with diameter at most $\epsilon \phi$. Therefore, for all $\phi \leq w$, $\alpha(m,n,w,\delta) \geq \gamma\left( \beta(m,w,\phi), n, \left\lceil \frac{\delta}{\phi} \right\rceil \right)$.
\end{lemma}

The proof of Lemma \ref{lemma:concatenation_anti} is similar to that of Proposition \ref{prop:concatenation} and is hence omitted. The lower bound on $\beta(n,w,\delta)$ in (\ref{eq:beta}) and the exact value of $\gamma(q,n,\delta)$ in (\ref{eq:gamma1}) can hence be used to compute lower bounds on $\alpha(m,n,w,\delta)$. We now give the Johnson bounds of type II for CWACs.

\begin{proposition}[Johnson bound of type II for CWACs] \label{prop:johnson_generalized}
For all $l|n$, $v \leq l$, $0 \leq \delta \leq \frac{n}{l}\min\{v,l-v\}$,
\begin{equation} \nonumber 
	A(m,n,w,d) \leq \frac{1}{\alpha \left(l,\frac{n}{l},v,\delta \right)} \left( \frac{m^l}{w^v(m-w)^{l-v}} \right)^{\frac{n}{l}} 
	A \left((m-1)l,\frac{n}{l},lw-v,d-\delta \right).
\end{equation}
\end{proposition}

\begin{proof}
Let $\mathcal{C}$ be an $(m,n,w)$ CWAC with minimum distance $d$ and cardinality $A(m,n,w,d)$ and let $\mathcal{L}$ be an antiCWAC in $J(l,v)^{\frac{n}{l}}$ with diameter $\delta$ and cardinality $\alpha(l,\frac{n}{l},v,\delta)$. For all ${\bf u} \in [m]^n$ and all ${\bf C} \in \mathcal{C}$, we define $f({\bf u}, {\bf C}) = (c_{0,u_0}, c_{1,u_1}, \ldots, c_{n-1,u_{n-1}}) \in \gf(2)^n$. Suppose $f({\bf u}, {\bf C}_0) , f({\bf u}, {\bf C}_1) \in \mathcal{L}$, then ${\bf C}'_0, {\bf C}'_1 \in J((m-1)l,lw-v)^{\frac{n}{l}}$, where ${\bf C}'_0$ and ${\bf C}'_1$ are the puncturings of the coordinates $(u_j,j)$ of ${\bf C}_0$ and ${\bf C}_1$, respectively. By definition of the Hamming distance, we obtain
\begin{equation} \label{eq:johnson_dh1}
	\dh({\bf C}_0, {\bf C}_1) = \dh({\bf C}'_0, {\bf C}'_1) + \dh(f({\bf u}, {\bf C}_0) , f({\bf u}, {\bf C}_1)).
\end{equation}
Since $f({\bf u}, {\bf C}_0) , f({\bf u}, {\bf C}_1) \in \mathcal{L}$, we have
$\dh(f({\bf u}, {\bf C}_0) , f({\bf u}, {\bf C}_1)) \leq 2\delta$ and (\ref{eq:johnson_dh1}) leads to $\dh({\bf C}'_0, {\bf C}'_1) \geq 2(d-\delta)$. 

For all ${\bf x} \in \gf(2)^n$ and all $\mathcal{D} \subseteq \gf(2)^n$, we let $\chi({\bf x}, \mathcal{D}) = 1$ if ${\bf x} \in \mathcal{D}$ and $\chi({\bf x}, \mathcal{D}) = 0$ otherwise. Hence for all ${\bf u} \in [m]^n$, $\sum_{{\bf C} \in \mathcal{C}} \chi(f({\bf u}, {\bf C}), \mathcal{L}) = |\{ {\bf C} \in \mathcal{C} : f({\bf u}, {\bf C}) \in \mathcal{L} \}|$ is the cardinality of a CWAC in $J((m-1)l,lw-v)^{\frac{n}{l}}$ with minimum distance at least $2(d-\delta)$, and hence 
\begin{equation} \label{eq:ja1}
	\sum_{{\bf C} \in \mathcal{C}} \chi(f({\bf u}, {\bf C}), \mathcal{L}) \leq A\left((m-1)l,\frac{n}{l}, lw-v,d-\delta\right)  
\end{equation}
for all ${\bf u} \in [m]^n$. Similarly, we have for all ${\bf C} \in \mathcal{C}$, $\sum_{{\bf u} \in [m]^n} \chi(f({\bf u}, {\bf C}), \mathcal{L}) = |\{ {\bf u} \in [m]^n : f({\bf u}, {\bf C}) \in \mathcal{L} \}|$. For any ${\bf L} \in J(l,v)^{\frac{n}{l}}$, it can be easily shown that $|\{ {\bf u} \in [m]^n : f({\bf u}, {\bf C}) = {\bf L}\}| = w^{\frac{nv}{l}} (m-w)^{n-\frac{nv}{l}}$, and hence
\begin{equation} \label{eq:ja2}
	\sum_{{\bf u} \in [m]^n} \chi(f({\bf u}, {\bf C}), \mathcal{L}) = \alpha \left(l,\frac{n}{l},v,2\delta \right) w^{\frac{nv}{l}} (m-w)^{n-\frac{nv}{l}}
\end{equation}
for all ${\bf C} \in \mathcal{C}$. Summing over all ${\bf u} \in [m]^n$ and all ${\bf C} \in \mathcal{C}$, while using (\ref{eq:ja1}) and (\ref{eq:ja2}), then finishes the proof.
\end{proof}

In practice, the bound in Proposition \ref{prop:johnson_generalized} needs to be minimized for all $l$, $v$, and $\delta$. However, in order to illustrate this bound, we display its values for the two extreme cases of $\delta$ in Corollary \ref{cor:johnson_generalized}. ???Vision of what these bounds mean: where do we puncture?

\begin{corollary} \label{cor:johnson_generalized}
We have
\begin{eqnarray} \nonumber 
 A(m,n,w,d) &\leq& \left(\frac{m}{w} \right)^{\frac{nv}{l}} \left(\frac{m}{m-w} \right)^{n-\frac{nv}{l}} A \left((m-1)l,\frac{n}{l},lw-v,d \right),\\
 \nonumber 
 A(m,n,w,d) &\leq& \left( \frac{m^l}{{l \choose v} w^v(m-w)^{l-v}} \right)^{\frac{n}{l}} A \left((m-1)l,\frac{n}{l},lw-v,d-\frac{n}{l}\min\{v,l-v\}\right).
\end{eqnarray}
\end{corollary}


We finish this section by deriving the counterparts of the Gilbert and Hamming bounds for CWACs. First, we study the number of vectors in $J(m,w)^n$ at distance $d$ from a given vector in $J(m,u)^n$. We denote the set of all nonnegative integer sequences of length $n$ whose sum is equal to $d$ as $P_n(d)$.

\begin{lemma} \label{lemma:N}
The number of words in $J(m,u)^n$ at distance $d$ from a given word in $J(m,w)^n$ is given by $N(m,n,w,u,d) = \sum_{\pi \in P_n(d)} \prod_{i \in \pi} \nu(i)$, where $\nu(i) = {w \choose \frac{u+w-i}{2}} {n-w \choose \frac{u+i-w}{2}}$.
\end{lemma}

\begin{proof}
The proof goes by induction on $n$. First, for $n=1$, $N(m,1,w,u,d)$ is given by $\nu(d)$, and hence the claim is true. Let us now assume it is true for $n = a$. It is easily shown that
\begin{eqnarray} \label{eq:N_recursion}
	N(m,a+1,w,u,d) &=& \sum_{e=0}^d N(m,a,w,u,e) \nu(d-e)\\
	\nonumber
	&=& \sum_{e=0}^d \nu(d-e) \sum_{\pi \in P_a(e)} \prod_{i \in \pi} \nu(i).
\end{eqnarray}
For any $\pi \in P_a(e)$, we have $\pi' = \pi \cup \{d-e\} \in P_{a+1}(d)$; conversely, any $\pi' \in P_{a+1}(d)$ can be expressed as $\pi \cup \{d-e\}$, where $\pi \in P_a(e)$. Therefore, we obtain $N(m,a+1,w,u,d) = \sum_{\pi' \in P_{a+1}(d)} \prod_{i \in \pi'} \nu(i)$.
\end{proof}

Although the value of the expression in Lemma \ref{lemma:N} cannot be computed in general, as it requires an exponential number of operations, it can be determined for relatively small values of $d$ and $n$. Also, (\ref{eq:N_recursion}) provides a recursive way to compute $N(m,n,w,u,d)$.

We now determine the counterparts of the Gilbert and Hamming bounds for CWACs, whose proofs are straightforward and hence omitted.

\begin{proposition}[Gilbert and Hamming bounds for CWACs] \label{prop:gilbert_hamming}
For all parameter values,
\begin{equation} \nonumber 
 \frac{{m \choose w}^n}{\sum_{i=0}^{d-1} N(m,n,w,w,2i)} \leq A(m,n,w,d) \leq 
 \min_{nw-d+1 \leq u \leq nw+d-1} \left\{ \frac{{m \choose u}^n}{\sum_{i=0}^{d-1} N(m,n,w,u,i)} \right\}.
\end{equation}
\end{proposition}

\subsection{Relations to other classes of codes} \label{sec:relations}

We first relate $A(m,n,w,d)$ to the maximum cardinalities of constant-weight codes and nonrestricted Hamming metric codes.  Proposition \ref{prop:trivial_relations} relates $A(m,n,w,d)$ to $B(n,w,d)$ and $C(q,n,d)$ for some particular values.

\begin{proposition} \label{prop:trivial_relations}
For all $d$, we have $A(m,n,1,d) = C(m,n,d)$ and $A(m,1,w,d) = B(m,w,d)$.
\end{proposition}

%

The proof is straightforward and hence omitted. Proposition \ref{prop:trivial_relations} indicates that the problem of finding the maximum cardinality of a CWAC generalizes both problems of finding optimal CWCs and optimal nonrestricted Hamming metric codes. Proposition \ref{prop:C<A<B}, on the other hand, views CWACs as a special class of constant-weight codes.

\begin{proposition}[CWACs as a subclass of constant-weight codes] \label{prop:C<A<B}
For all $a$, $m'$, and $w'$ such that $am' \leq m$ and $aw' \leq w \leq aw' + m-aw'$, we have $A(m,n,w,d) \geq A \left( m',an,w', d\right)$. Therefore, for all $w \leq m$, we have
\begin{equation} \label{eq:C<A<B}
 C \left( \left\lfloor \frac{m}{w} \right\rfloor, nw, d \right) \leq A(m,n,w,d) \leq B(mn,nw,d).
\end{equation}
\end{proposition}

\begin{proof}
For all ${\bf C} \in J(m',w')^{an}$, let the $k$-th column of $f({\bf C}) \in J(m,w)^n$ consist of the columns $ak, ak+1, \ldots, (a+1)k-1$ of ${\bf C}$, of $w-aw'$ ones, and of $m-am'-w+aw'$ zeros. Then $f$ is an injection from $J(m',w')^{an}$ to $J(m,w)^n$ which preserves the distance. We hence obtain the first claim; (\ref{eq:C<A<B}) follows immediately.
\end{proof}

Proposition \ref{prop:A_upper_bound} refines the upper bound in (\ref{eq:C<A<B}) by providing an explicit extension of a CWAC by adding several CWACs of different dimensions.

\begin{proposition}[Refined upper bound for CWACs] \label{prop:A_upper_bound}
For all $l \geq 1$ and any decreasing sequence of integers $\{n_0 = n, n_1, \ldots, n_{l-1}\}$ satisfying $n_i | nw$ and $\frac{nw}{m} \leq n_i \leq \left(1- \frac{d}{nw} \right) n_{i-1}$ for all $1 \leq i \leq l-1$, we have
\begin{equation} \nonumber 
 \sum_{i=0}^{l-1} B(n,n_i,d) A(m,n_i,w_i,d) \leq B(mn,nw,d),
\end{equation}
where $w_i = \frac{nw}{n_i}$.
\end{proposition}

\begin{proof}
We shall construct a constant-weight code in $J(mn,nw)$ with minimum distance $d$ and cardinality $\sum_{i=0}^{l-1} B(n,n_i,d) A(m,n_i,w_i,d)$. For all $i$, let $\mathcal{C}^i = \{ {\bf C}^{i,k} \}$ be a CWAC in $J(m,w_i)^{n_i}$ with minimum distance $d$ and cardinality $A(m,n_i,w_i,d)$ and let $\mathcal{L}^i = \{L^{i,j}\}$ be a constant-weight code in $J(n,n_i)$ with minimum distance $d$ and cardinality $B(n,n_i,d)$. Also for any $i,j$, we denote the nonzero coordinates of $L^{i,j}$ as $L^{i,j}(a)$ for $0 \leq a \leq n_i-1$. Then for all $k$, we construct the vector ${\bf X}^{i,j,k}$ as follows. For all $0 \leq a \leq n_i-1$, the $L^{i,j}(a)$-th column of ${\bf X}^{i,j,k}$ is given by the $a$-th column of ${\bf C}^{i,k}$, and all the other $n-n_i$ columns are set to zero. Then the code $\mathcal{C} = \left\{ {\bf X}^{i,j,k} \right\}$ is a constant-weight code in $J(mn,nw)$. 

We now show that $\mathcal{C}$ has minimum distance $d$ by considering two distinct codewords ${\bf X}^{i,j,k}$ and ${\bf X}^{i',j',k'}$ in $\mathcal{C}$. First, if $i \neq i'$, without loss of generality $i < i'$ and hence $n_i > n_{i'}$, we have $\dh({\bf X}^{i,j,k}, {\bf X}^{i',j',k'}) = \sum_{a = 0}^{n-1} \dh({\bf x}_a^{i,j,k}, {\bf x}_a^{i',j',k'})$, where ${\bf x}_a$ denotes the $a$-th column of ${\bf X}$. Therefore,
\begin{eqnarray}
 \nonumber
 \dh({\bf X}^{i,j,k}, {\bf X}^{i',j',k'}) &=& \sum_{a \in L^{i,j} \cap L^{i',j'}} \dh({\bf x}_a^{i,j,k}, {\bf x}_a^{i',j',k'}) + 
 \sum_{a \in L^{i,j} \backslash L^{i',j'}} \dh({\bf x}_a^{i,j,k}, {\bf 0}) + \sum_{a \in L^{i',j'} \backslash L^{i,j}} \dh({\bf 0}, {\bf x}_a^{i,j,k})\\
 \nonumber
 &\geq& \left| L^{i,j} \cap L^{i',j'}\right| (w_{i'} - w_i) + \left| L^{i,j} \backslash L^{i',j'} \right| w_i + \left| L^{i',j'} \backslash L^{i,j} \right| w_{i'}\\
 \nonumber
 &=& 2w_i\left(n_i- \left| L^{i,j} \cap L^{i',j'} \right| \right)\\
 \nonumber
 &\geq& 2w_i(n_i - n_{i'})\\
 \label{eq:A_ub1}
 &\geq& 2d,
\end{eqnarray}
where (\ref{eq:A_ub1}) follows the definition of the $\{n_i\}$ sequence. Second, if $i=i'$ and $j \neq j'$, then $\dh({\bf X}^{i,j,k}, {\bf X}^{i,j',k'}) \geq \dh(L^{i,j}, L^{i,j'}) \geq 2d$ by the minimum distance of $\mathcal{L}^i$. Third, if $i=i'$, $j=j'$, and $k \neq k'$, then $\dh({\bf X}^{i,j,k}, {\bf X}^{i,j,k'}) = \dh({\bf C}^{i,k}, {\bf C}^{i,k'}) \geq 2d$ by the minimum distance of $\mathcal{C}^i$. Thus $\mathcal{C}$ has minimum distance $2d$ and cardinality $\sum_{i=0}^{l-1} B(n,n_i,d) A(m,n_i,w_i,d) \leq B(mn,nw,d)$.
\end{proof}

Proposition \ref{prop:A_lower_bound} below gives another means to tighten the upper bound in Proposition \ref{prop:C<A<B}.

\begin{proposition} \label{prop:A_lower_bound}
For any $a \geq 1$, we have $A(am,n,w,d) \geq C \left(a,n,\left\lceil \frac{d}{w} \right\rceil \right) A(m,n,w,d)$. Furthermore, if $a|w$ and $d \leq nw \frac{a-1}{a}$, then
\begin{equation} \nonumber 
	A(am,n,w,d) \geq C \left(a,n,\left\lceil \frac{d}{w} \right\rceil \right) A(m,n,w,d) + A\left( m,an,\frac{w}{a},d \right).
\end{equation}
\end{proposition}

\begin{proof}
Let $\mathcal{C}$ be a CWAC in $J(m,w)^n$ with minimum distance $d$ and cardinality $A(m,n,w,d)$ and let $\mathcal{V}$ be a code in $[a]^n$ with minimum distance $\left\lceil \frac{d}{w} \right\rceil$ and cardinality $C \left(a,n,\left\lceil \frac{d}{w} \right\rceil \right)$. We construct the code $\mathcal{D} \subseteq J(am,w)^n$ as follows. For any ${\bf C}^i \in \mathcal{C}$, ${\bf v}^j \in \mathcal{V}$, and all $0 \leq k \leq n-1$, let ${\bf d}^{i,j}_k$ be the vector in $\gf(2)^{am}$ with $a$ columns of length $m$ where only the $v^j_k$-th column is nonzero and given by the $k$-th column of ${\bf C}^i$. Then the codeword ${\bf D}^{i,j} \in \mathcal{D}$ is defined to be ${\bf D}^{i,j} = ({\bf d}^{i,j}_0, {\bf d}^{i,j}_1, \ldots, {\bf d}^{i,j}_{n-1})$.

We now show that $\mathcal{D}$ has minimum distance at least $d$ by considering ${\bf D}^{i,j}$ and ${\bf D}^{i',j'}$. First, for any pair of codewords, we have 
\begin{equation} \label{eq:Alb1}
	\dh({\bf D}^{i,j}, {\bf D}^{i',j'}) = \sum_{k=0}^{n-1} \dh({\bf d}^{i,j}_k, {\bf d}^{i',j'}_k)
	= 2w \dh({\bf v}_j, {\bf v}_{j'}) + \sum_{k:v^k_j = v^k_{j'}} \dh({\bf C}^i_k, {\bf C}^{i'}_k).
\end{equation}
If $i \neq i'$, (\ref{eq:Alb1}) leads to $\dh({\bf D}^{i,j}, {\bf D}^{i',j'}) \geq \dh({\bf C}^i, {\bf C}^{i'}) \geq 2d$. Also, if $i=i'$ and $j \neq j'$, (\ref{eq:Alb1}) yields $\dh({\bf D}^{i,j}, {\bf D}^{i',j'}) \geq 2w \dh({\bf v}_j, {\bf v}_{j'}) \geq 2d$. Therefore $\mathcal{D}$ has minimum distance at least $2d$. As a corollary, the cardinality of $\mathcal{D}$ is given by $|\mathcal{D}| = |\mathcal{C}||\mathcal{V}|$.

We now suppose $a|w$ and $d \leq nw \frac{a-1}{a}$, and we let $\mathcal{E}$ be a CWAC in $J(m,\frac{w}{a})^{an}$ with minimum distance $2d$ and cardinality $A(m,an,\frac{w}{a},d)$. Since any codeword in $\mathcal{D}$ has $(a-1)n$ all-zero columns, and that the remaining ones have weight $w$, we have $\dh(\mathcal{E},\mathcal{D}) \geq (a-1)n\frac{w}{a} + n\left(w-\frac{w}{a}\right) \geq 2d$. Thus the code $\mathcal{E} \cup \mathcal{D}$ has minimum distance $2d$ and cardinality $|\mathcal{E}|+ |\mathcal{D}|$.
\end{proof}

We also determine a close relation between CWACs of high minimum distance and constant-weight codes.

\begin{proposition} \label{prop:A<B}
For all $m$, $n$, $w$, and $d = (n-1)w + r$ with $0 < r \leq w$, $A(m,n,w,d) \leq B(m,w,r)$.
\end{proposition}

\begin{proof}
Let $\mathcal{C}$ be a CWAC in $J(m,w)^n$ with minimum distance $d > (n-1)w$ and cardinality $A(m,n,w,d)$. For all ${\bf C}^0, {\bf C}^1 \in \mathcal{C}$ and all $0 \leq i \leq n-1$, we have $\dh({\bf C}^0_i, {\bf C}^1_i) \geq 2r$ and hence $\mathcal{C}_i$ is a constant-weight code in $J(m,w)$ with minimum distance no less than $2r$ and cardinality $|\mathcal{C}|$. Therefore, we have $A(m,n,w,d) \leq B(m,w,r)$.
\end{proof}

The following upper bound can be viewed as the counterpart of the concatenation construction in Section \ref{sec:concatenation}.

\begin{proposition} \label{prop:A<C}
For all $m$, $n$, $w$, and $d$, $A(m,n,w,d) \leq C \left({m \choose w}, n, \left\lceil \frac{d}{w} \right\rceil \right)$.
\end{proposition}

\begin{proof}
Let $\mathcal{C} \in J(m,w)^n$ have minimum distance $d$ and cardinality $A(m,n,w,d)$. There is a bijection between $J(m,w)$ and ${m \choose w}$; applying it columnwise, we map $\mathcal{C}$ to a code $f(\mathcal{C}) \subseteq [{m \choose w}]^n$. For any ${\bf C}^0, {\bf C}^1 \in \mathcal{C}$, we have 
\begin{equation} \nonumber
	d \leq	\dh({\bf C}^0, {\bf C}^1) \leq 2w |\{i : {\bf C}^0_i \neq {\bf C}^1_i \}| = 2w \dh(f({\bf C}^0), f({\bf C}^1)),
\end{equation}
and hence $f(\mathcal{C})$ has minimum distance at least $\frac{d}{w}$.
\end{proof}

The Bassalygo-Elias bound \cite{Bas65} is a crucial relation between the maximum cardinalities of constant-weight codes and nonrestricted binary codes. We now derive a generalization of the Bassalygo-Elias bound for CWACs.

\begin{proposition}[Bassalygo-Elias bound for CWACs] \label{prop:bassalygo}
For all $w_1 | w_2 | m$, we have
\begin{equation} \label{eq:bassalygo_w1}
 A \left(m\frac{w_1}{w_2}, n \frac{w_2}{w_1}, w_1, d\right) \geq \left\lceil \frac{ {m \frac{w_1}{w_2} \choose w_1}^{n\frac{w_2}{w_1}} }{ {m \choose w_2}^n }  A(m,n,w_2,d) \right\rceil.
\end{equation}
Therefore,
\begin{equation} 
 \label{eq:bassalygo1}
 \frac{{m \choose w}^n}{{mn \choose wn}} B(mn,nw,d) \leq A(m,n,w,d) \leq \frac{{m \choose w}^n}{\left(\frac{m}{w}\right)^{nw}} C\left(\frac{m}{w},nw,d \right).
\end{equation}
\end{proposition}

\begin{proof}
Let $\mathcal{C}$  be a CWAC in $J(m,w_2)^n$ with minimum distance $2d$ and cardinality $A(m,n,w_2,d)$. Let $\mathfrak{S}_m$ denote the symmetric group on $m$ elements and let $\mathfrak{S}_m^n$ be its $n$-fold cartesian product. For all $\pi \in \mathfrak{S}_m^n$, ${\bf C} \in \mathcal{C}$, we define the function $f(\pi,{\bf C})$ to be $f(\pi,{\bf C}) = 1$ if $\pi({\bf C}) \in J \left(m\frac{w_1}{w_2}, w_1 \right)^{n \frac{w_2}{w_1}}$ and $f(\pi,{\bf C}) = 0$ otherwise. We have for all $\pi \in \mathfrak{S}_m^n$
\begin{equation} \label{eq:b1}
 \sum_{{\bf C} \in \mathcal{C}} f(\pi,{\bf C}) = \left| \left\{ \pi({\bf C}) \in J\left(m\frac{w_1}{w_2}, w_1\right)^{n \frac{w_2}{w_1}} : {\bf C} \in C \right\} \right| \leq A\left(m\frac{w_1}{w_2}, n \frac{w_2}{w_1}, w_1, 2d\right)
\end{equation}
and for all ${\bf C} \in \mathcal{C}$
\begin{equation} \label{eq:b2}
 \sum_{\pi \in \mathfrak{S}_m^n} f(\pi,{\bf C}) = \left\{ {m \frac{w_1}{w_2} \choose w_1}^{\frac{w_2}{w_1}} w_2! (m-w_2)! \right\}^n.
\end{equation}
Combining (\ref{eq:b1}) and (\ref{eq:b2}), the double sum leads to (\ref{eq:bassalygo_w1}). In particular, applying (\ref{eq:bassalygo_w1}) with a simple change of variables leads to (\ref{eq:bassalygo1}). 
\end{proof}

Using a well-known property of optimal binary codes, the Bassalygo-Elias bound was refined by van Pul (see \cite{AGM92}) by a factor of $2$. However, no analogous property is known for optimal constant-weight codes so far, hence such a refinement cannot be performed. Corollary \ref{cor:bassalygo_log} below gives a slight loosening of the bound in (\ref{eq:bassalygo1}) which is much easier to compute for large parameter values.

\begin{corollary} \label{cor:bassalygo_log}
For $2w \leq m$, we have $A(m,n,w,d) \geq \sqrt{\pi w n} 2^{-\frac{n}{2}(3 +\log w)} B(mn,nw,d)$.
\end{corollary}

\begin{proof}
Applying the bounds on the binomial coefficient in \cite[Ch. 10, Lemma 7]{MS77} to (\ref{eq:bassalygo1}), we obtain $A(m,n,w,d) \geq \sqrt{\frac{2\pi wn(1-\frac{w}{m})}{(8w(1-\frac{w}{m}))^n}} B(mn,nw,d)$. Since $w \leq \frac{m}{2}$, we easily obtain the desired result.
\end{proof}


Binary codes, and especially constant-weight codes, are closely related to spherical codes. More precisely, any constant-weight code can be mapped into a spherical code of a given length \cite[Theorem 2]{AVZ00}. The cardinality of a constant-weight code is hence upper bounded by the maximum cardinality of an optimal spherical code with related dimension and maximum cosine. Proposition \ref{prop:spherical} generalizes this relation to the case of CWACs.

\begin{proposition}[Relation between CWACs and spherical codes] \label{prop:spherical}
We have $$A(m,n,w,d) \leq \As\left(n(m-1), 1 - \frac{dm}{nw(m-w)}\right).$$
\end{proposition}

\begin{proof}
Let $\Omega$ denote the mapping from $\gf(2)^{mn}$ to $\mathbb{R}^{mn}$ which sends $0 \rightarrow 1$ and $1 \rightarrow -1$ coordinatewise. Then $\Omega(J(m,w)^n) = \{ {\bf x} \in \{1,-1\}^{mn} : {\bf x} {\bf A} = (m-2w){\bf 1} \}$, where ${\bf A} \in \mathbb{R}^{mn \times n}$ is formed by cascading $n$ all-ones column vectors of length $m$. Hence any point ${\bf x} \in \Omega(J(m,w)^n)$ satisfies $({\bf x} - {\bf c}){\bf A} = {\bf 0}$, where ${\bf c} = \frac{m-2w}{m}{\bf 1}$ and $|| {\bf x} - {\bf c} || = 2 \sqrt{\frac{nw(m-w)}{m}} = r$. Therefore, $\Omega(J(m,w)^n)$ is a subset of the $n(m-1)$-dimensional sphere of radius $r$ around ${\bf c}$. Also, it is easily shown that $d_E(\Omega({\bf X}), \Omega({\bf Y})) = 2 \sqrt{\dh({\bf X}, {\bf Y})}$ for all ${\bf X}, {\bf Y} \in \gf(2)^{mn}$. Let $\mathcal{C}$ be a CWAC in $J(m,w)^n$ with minimum distance $2d$. Then the code $\frac{1}{r}(\Omega(\mathcal{C}) - {\bf c})$ is a spherical code of dimension $n(m-1)$ with minimum Euclidean distance $2\sqrt{2d}$ and maximum cosine given by $1 - \frac{4d}{r^2}$.
\end{proof}

We remark that once the total length $M = mn$ and the total weight $W = nw$ of the CWAC are fixed, the maximum cosine in the upper bound of Proposition \ref{prop:spherical} remains constant, while the dimension of the space is a decreasing function of $n$. This dimension reaches its maximum for $n = 1$, where Proposition \ref{prop:spherical} corresponds to \cite[Theorem 2]{AVZ00} established for constant-weight codes and reaches its minimum value for liftings of Hamming metric codes.

Using the exact values of $\As(n,s)$ reviewed in Section \ref{sec:preliminaries}, Proposition \ref{prop:spherical} can be reduced as follows
\begin{eqnarray}
	\nonumber
	A(m,n,w,d) &\leq& \left\lfloor \frac{d}{d - \Delta} \right\rfloor, \quad \text{if } d \geq \Delta\left(1 + \frac{1}{n(m-1)}\right)\\
	\nonumber
	A(m,n,w,d) &\leq& n(m-1) + 1, \quad \text{if } \Delta < d \leq \Delta\left(1 + \frac{1}{n(m-1)}\right)\\
	\nonumber
	A\left(m,n,w, \Delta \right) &\leq& 2n(m-1),
\end{eqnarray}
where $\Delta = \frac{nw(m-w)}{m}$

\end{document}